\documentclass[11pt]{article}

\usepackage[letterpaper,margin=1in]{geometry}
\usepackage{setspace}
\usepackage{amsthm,amsmath,amssymb}
\usepackage{tcolorbox}
\usepackage{graphicx}
\usepackage{bm}
\usepackage[ruled,vlined,linesnumbered]{algorithm2e}
\usepackage{amsfonts}
\usepackage{mathtools}
\tcbuselibrary{skins}
\usepackage{xcolor}
\usepackage{cancel}
\usepackage{enumitem}
\usepackage{csquotes}
\usepackage[]{mdframed}
\usepackage{caption}
\usepackage[dvipsnames]{xcolor}
\usepackage{float}
\usepackage[bb=boondox]{mathalfa}
\usepackage{tikz}
\usetikzlibrary{patterns}
\usepackage[colorinlistoftodos]{todonotes}
\usepackage{authblk}
\colorlet{mix}{red!50!black}
\definecolor{DarkRed}{rgb}{0.8,0,0}
\definecolor{ForestGreen}{rgb}{0.1333,0.5451,0.1333}
\usepackage{hyperref}
\hypersetup{
  colorlinks=true,
  linkcolor=black,   
  citecolor=ForestGreen,
  urlcolor=DarkRed
}
\usepackage{cleveref} 
\usepackage{subcaption}
\theoremstyle{plain}
\newtheorem{theorem}{Theorem}

\newtheorem{lemma}{Lemma}

\theoremstyle{definition}
\newtheorem{definition}{Definition}

\newtheorem{observation}{Observation}

\theoremstyle{remark}

\crefname{theorem}{Theorem}{theorems}
\Crefname{theorem}{Theorem}{Theorems}
\crefname{result}{Result}{results}
\Crefname{result}{Result}{Results}
\crefname{lemma}{Lemma}{lemmas}
\Crefname{lemma}{Lemma}{Lemmas}
\crefname{sublemma}{Sublemma}{lemmas}
\Crefname{sublemma}{Sublemma}{Lemmas}
\crefname{corollary}{corollary}{corollaries}
\Crefname{corollary}{Corollary}{Corollaries}
\crefname{proposition}{proposition}{propositions}
\Crefname{proposition}{Proposition}{Propositions}
\crefname{claim}{claim}{claims}
\Crefname{claim}{Claim}{Claims}
\crefname{observation}{observation}{observations}
\Crefname{observation}{Observation}{Observations}
\crefname{definition}{definition}{definitions}
\Crefname{definition}{Definition}{Definitions}
\crefname{fact}{fact}{facts}
\Crefname{fact}{Fact}{Facts}
\crefname{example}{example}{examples}
\Crefname{example}{Example}{Examples}
\crefname{remark}{remark}{remarks}
\Crefname{remark}{Remark}{Remarks}

\newboolean{shortver}
\setboolean{shortver}{false}

\newcommand{\frd}{Fr\'echet distance~}
\newcommand{\frdl}{Fr\'echet distance}

\newcommand{\klcp}{$(k, \ell)-$center problem~}

\newcommand{\lab}{\ensuremath{\ell(a,b)}~}
\newcommand{\lax}{\ensuremath{L_{x_1,x_2,\ldots x_{2(d-1)}}~}}
\newcommand{\U}[1]{\ensuremath{\mathcal{U}({#1})}~}

\newcommand{\df}[2]{\ensuremath{\delta_F({#1},{#2})}}
\newcommand{\mtc}[1]{\ensuremath{\mathcal{#1}}}
\newcommand{\dhl}[1]{\delta_{\vec{h} \mtc{L}}({#1})}
\newcommand{\dhr}[1]{\delta_{\vec{h} \mtc{R}}({#1})}
\newcommand{\dhm}[1]{\delta_{\vec{h} \mtc{M}}({#1})}

\definecolor{bblue}{rgb}{0,0.1,0.55}

\definecolor{defblue}{rgb}{0.1,0.4,0.6}
\let\emph\relax\DeclareTextFontCommand{\emph}{\color{bblue}\em}

\title{Polynomial-time algorithm for exact $(1,2)$-center\\ problem under continuous \frd}

\author{
Soumya Bhattacharya\thanks{Indian Statistical Institute, Kolkata, India. Email: soumyakolkata22@gmail.com}
\hspace{1cm}
Serene Rasheed\thanks{Indian Institute of Technology Delhi, New Delhi, India. Email: serenerasheed1@gmail.com}
\hspace{1cm}
Sasanka Roy\thanks{Indian Statistical Institute, Kolkata, India. Email: sasanka.ro@gmail.com}
}

\date{}
\begin{document}
\maketitle
\begin{abstract}
    In this paper, we explore the $(1,2)$-center problem for polygonal curves under continuous \frdl. The $(k,\ell)$-center problem, in general, is known to be NP-hard. Aronov, Filtser, Horton, Katz, and Sheikhan (WADS'19) gave a polynomial-time algorithm for the $(1,2)$-center of curves in the plane under the discrete \frd. To the best of our knowledge, the $(1,2)$-center under continuous \frd has not been studied yet. We present a polynomial time algorithm to solve the problem exactly under both $\mathbb{L}_2$ and $\mathbb{L}_\infty$ norm, running in $O\bigr((n^2r+nr^2)^{2+\epsilon}\bigl)$ time for curves in the plane where $r$ is the number of input curves and $n$ is the maximum complexity of any curve. Further, for curves in any dimension $d$, the expected time to compute the center using the algorithm is $O\bigr((n^2r+nr^2)^{2(d-1)+\epsilon}\bigl)$. We have also shown that an $(1+\widetilde{\epsilon})-$factor approximation of $(1,2)$-center can be computed in $O(n^2r+nr^2+1/\epsilon^s)$ time for any $\epsilon>\widetilde{\epsilon}>0$ and some constant $s$ for curves in the plane. For curves in the plane, we have shown that, with the center restricted to be horizontal, we can compute the exact center in $O(n^2r+nr^2)$ time. An algorithm has been introduced to find a $3$-factor approximation of the $(1,2)$-center in time linear in the number of curves. A formulation was introduced by de Berg, Mehrabi, and Ophelders (CCCG'17) to measure \frd between a curve and a query segment under $\mathbb{L}_2$ norm for curves in the plane. We have shown the formulation is valid under both $\mathbb{L}_2$ and $\mathbb{L}_\infty$ norm for curves in $\mathbb{R}^d$.
\end{abstract}

\section{Introduction}
Measuring similarity among polygonal curves plays an important role in many fields, including the analysis of the movement of GPS-enabled objects, biomedical science, river networks, and sports analytics. Similarity measures between geometric objects, such as polygonal curves, require metrics that respect their inherent structure; in this regard, the Fréchet distance has emerged as particularly well-suited. Unlike simpler metrics, the Fréchet distance accounts for the ordering and continuity of points along a curve, making it especially effective for capturing geometric resemblance.

Researchers have used the Fréchet distance to study the structure of biomolecules such as proteins \cite{app_biomedical_jiang2008protein}, analyze animal movement \cite{app_anm_mv_Cleasby2019-al}, perform handwriting matching \cite{app_hand_writing_sriraghavendra2007frechet}, and analyze sports trajectories \cite{app_sports_tang2017efficient}, among many other real-world applications involving similarity between polygonal curves.
\paragraph*{Background:} A polygonal curve of complexity $n$, is defined as the linear interpolation of a sequence of $n$ points. Authors in \cite{AltG95} introduced an algorithm for computing the Fréchet distance between two polygonal curves in $O(mn\log mn)$ time, where $m,n$ are the complexities of the curves, and formally established the appropriateness of \frd for measuring curve similarity. Authors in \cite{ChengHuang2025} recently improved that result. Simplification of a curve and clustering set of curves under \frd has been studied extensively \cite{simplification_Agarwal,ChengHuang2023, ChengHuangJiang2025, vanKreveldLofflerWiratma2018, BringmannChaudhury2019, vanderHoogEtAl2025, AronovFarhanaKatzRamesh2025}. Authors in \cite{data_structure_frd, buchin_eff_frech} worked on the problem to prepare a data structure for a given curve so that for a specific query segment, one can answer the \frd between the segment and the curve quickly. There have been many studies regarding preparing a data structure for a given set of curves, so that for any query curve, one can find the nearest curve from the query curve in a shorter time \cite{ChengHuang2024, discrete_12}. Another problem of interest in this regard is $(k,\ell)$-median/center clustering. Given a set of input curves, $(k,\ell)$-median problem asks to find $k$ median curves of complexity at most $\ell$ that minimize the sum of \frd over all input curves to their closest median curve. Authors in \cite{BuchinDriemelRohde2023} worked on this problem. On the other hand, the $(k,\ell)$-center problem asks to find $k$ center curves of complexity at most $\ell$ that minimize the \frd over all input curves to their closest center curve. To define formally:

\textbf{$\boldsymbol{(k,\ell)}$-center problem:} Given a set of polygonal curves $\pi = \{\pi_1, \pi_2, \ldots, \pi_r\}$ in $\mathbb{R}^d$, the $(k,\ell)$-center problem asks for a set of $k$ polygonal curves (centers)- $c = \{c_1, c_2, \ldots, c_k\}$ in $\mathbb{R}^d$, not necessarily chosen from $\pi$, each of complexity at most $\ell$, minimizing $\max_{\pi_i \in \pi} \; \min_{c_j \in c} \; \delta_F(\pi_i, c_j)$. This minimum value is called the distance of the centers to the set of input curves, and a set of centers $c$ that attains it is called an optimal center set.

The \klcp under discrete \frd for time series (univariate polygonal curves) was studied in \cite{univariate_clustering}. There, the problem was proved to be NP-hard, and a $(1+\epsilon)$-factor approximation algorithm with running time $O(rn\log r)$ was obtained, where $r,n$ are the number of curves and the maximum complexity of any curve in the input set, for any $\epsilon>0$ and constant $k, \ell$. The $(k,\ell)$-center problem for general polygonal curves under both discrete and continuous \frd was studied in \cite{kl_center}, where it was shown that, under the continuous case, the problem is NP-hard to approximate within any factor smaller than $2.25$, and a $3$-factor approximation algorithm was given, running in $O\bigl(kn(r\ell \log(\ell+n)+n^2\log n)\bigr)$. The particular case $k=1,\ell=2$ was studied in \cite{discrete_12} for curves in the plane, and a polynomial-time algorithm for the $(1,2)$-center problem under discrete \frd was obtained, but the continuous case remained untouched.

\paragraph*{Our Results:} We observed that the particular case $k=1,\ell=2$, can be solved in polynomial time even under continuous \frd for curves in the plane. For any $\epsilon>0$, we obtained an $O\bigl((n^2r+nr^2)^{2+\epsilon}\bigr)$ time algorithm for curves in $\mathbb{R}^2$  to find $(1,2)$ center where $r$ is the number of curves and $n$ is the maximum complexity of any curve in the input set. For curves in any dimension $d$, the expected time to compute the center is $O\bigr((n^2r+nr^2)^{2(d-1)+\epsilon}\bigl)$.  We have also shown that, for curves in the plane with the center restricted to be horizontal, we can solve the problem in $O(n^2r+nr^2)$ time.

Along with results for exact $(1,2)$-center we obtained some results for some approximation schemes; we have shown that one can compute an $(1+\widetilde{\epsilon})-$factor approximation of $(1,2)$-center in $O(n^2r+nr^2+1/\epsilon^s)$ time for any $\epsilon>\widetilde{\epsilon}>0$ and some constant $s$. We also obtained a simple algorithm for curves in the plane that runs in $O(r)$ time and computes a $3$-factor approximation of the $(1,2)$-center.

All our results are valid under both $\mathbb{L}_2$ and $\mathbb{L}_\infty$ norm. Further, we extend a formulation of \cite{data_structure_frd}, originally given for measuring \frd between a curve and a query segment under $\mathbb{L}_2$ norm for curves in the plane, to both $\mathbb{L}_2$ and $\mathbb{L}_\infty$ norms and to curves in $\mathbb{R}^d$.

\section{Preliminaries}\label{sec_prelim}
A polygonal curve $P$ is the linear interpolation of a sequence of points. These points are called the \emph{vertices} of $P$, and the number of such vertices is the \emph{complexity} of the curve. We represent a polygonal curve as a function $P : [0, n] \rightarrow \mathbb{R}^d$, where $p_{i + k} = (1-k)\cdot p_i + k \cdot p_{i+1},\: \text{for } i \in \{0, \ldots, n-1\},\ k \in [0,1]$. Let $P$ and $Q$ be two polygonal curves of complexity $n$ and $m$, respectively, i.e., $P : [0, n-1] \rightarrow \mathbb{R}^d$ and $Q : [0, m-1] \rightarrow \mathbb{R}^d$. The \emph{Fr\'echet distance} between $P$ and $Q$ is defined as:

\begin{equation}\label{eq_frd_alt_gadou}
D_F(P, Q) = \inf_{\alpha, \beta} \; \sup_{t \in [0,1]} \; d\bigl(P(\alpha(t)), Q(\beta(t))\bigr),
\end{equation}

where $\alpha : [0,1] \rightarrow [0, n-1]$ and $\beta : [0,1] \rightarrow [0, m-1]$ range over all continuous, non-decreasing, surjective (reparametrization) functions, and $d(\cdot,\cdot)$ denotes the distance under the concerned norm.

We will now introduce some notation and definitions that will be used or referred to throughout the paper. \lax will denote the line in $\mathbb{R}^d$ with parameter value $x_1,x_2,\ldots,x_{2(d-1)}$. Throughout this paper, mention of $\lax$ or simply $L$ indicates a directed line. For a point $p$ and a line $L$, $p^{*}(L)$ denote the point on $L$ nearest to $p$ in the distance norm under consideration; when this nearest point is not unique, i.e. in the case of an axis-parallel line when distance is measured under $\mathbb{L}_\infty$ norm, $p^{*}(L)$ denotes the closed, connected set of nearest points. For two points $p_i,p_j$, if $p_i^*(L)$ appears after $p_j^*(L)$ while traversing the line $L$ in its direction, we write $p_i^*(L)>_{L} p_j^*(L)$ (when $p_i^*(L)$ or $p_j^*$(L) are not unique points, we compare where $p_i^*(L)$ or $p_j^*$(L) starts to appear). When $L$ is clear from the context, we simply write $p_i^*$ or $p_i^*> p_j^*$ instead of $p_i^*(L)$ or $p_i^*>_L p_j^*$. We write $\lab$ to represent a line segment with endpoints $a,b$ and whenever we mention $\lab \in L$, it is assumed that $a\leq_Lb$.

\begin{definition}[Backpair]
 We call a pair of vertices $(p_i, p_j)$, belonging to the same curve, a backpair with respect to a line $L$ if $i<j$ but $p_i^*>_{L} p_j^*$. 
\end{definition}

\begin{definition}
    Let $(p_i, p_j)$ be a backpair with respect to $L$. The optimal point $q$ on $L$ for which $\max(||p_i-q||,||p_j-q||)$ is minimum is called \emph{optimal point on}\footnote{when this optimal point is not unique, we consider the closed, connected set of optimal points} $L$ for the backpair-$(p_i, p_j)$. This minimized maximum distance is called \emph{backpair distance}.
\end{definition}

$$B_{(p_i,p_j)}(L) =  \min_{q \in L} \max \{ ||p_i - q|| ,||p_j - q||  \} $$

Backpair distance between a line segment and a polygonal curve is defined as:
$$ B(P,L) = \max_{\substack{ \forall p_i, p_j \in P \\ i < j\\ p_i^* >_{L} p^*_j }}B_{(p_i,p_j)}(L)$$

For an illustration of the optimal point due to a backpair check \cref{fig_optimal_bpoint}.

\begin{observation}\label{fac_opt_midl}
    For a backpair $p_i,p_j$ with respect to $L$, and its optimal point $q$ on $L$ satisfies the relation: $p_i^*\ge_Lq\ge_L p_j^*$.
\end{observation}

\begin{definition}[Hausdorff distance]
    Given two curves, for each point on either curve, consider its distance to the nearest point on the other curve. The maximum of all these distances is called the Hausdorff distance between the curves.
\end{definition}

$$\delta_{\overrightarrow{H}}(P,Q)=\max\{\max_{\substack{p \in P}} \min_{\substack{q \in Q}}||p-q||,\max_{\substack{q \in Q}} \min_{\substack{p \in P}}||p-q||\}$$

Let us define some distance function for a polygonal curve $P$ and a line segment $\lab\in L$ :
$$ \dhl{P, \lab}\coloneq \max_{\substack{p_i \in P\\p_i^* <_L a}} ||a - p_i|| ,$$ 
$$ \dhr{P, \lab}= \max_{\substack{p_i \in P\\p_i^* >_L b}} ||b - p_i|| ,$$ 
$$ \dhm{P, L}= \max_{\substack{ p_i \in P\\a\leq_Lp_i^*\leq_L b}} ||p_i-p_i^*(L)|| ,$$

Lemma $5$ of \cite{data_structure_frd} shows the Hausdorff distance between $P$ and \lab can be formulated as:
$$\delta_{\overrightarrow{H}}(P,\lab)=\max\bigl(\dhl{P, \lab},\dhr{P, \lab},\dhm{P, L}\bigr)$$

\section{Measuring \frd between a polygonal curve and a line segment in any fixed dimension}
In \cite{data_structure_frd} \frd between a polygonal curve and a line segment in $\mathbb{R}^2$, restricting the line segment to be horizontal under $\mathbb{L}_2$ norm has been formulated. In our notation, this formulation is as follows.

\begin{equation}\label{Berg_eq}
    \begin{aligned}
        \df{P}{\lab} = \max\{ ||p_0 - a|| , ||p_n - b||, \dhl{P,\lab},\dhr{P,\lab},\\ \dhm{P,L}, B(P,L) \}
    \end{aligned}
\end{equation}

where $p_0,p_n$ are the start and end vertices of the polygonal curve $P$; $L$ is a horizontal line and $\lab \in L$. 

In \cite{buchin_eff_frech} it was shown that this formulation can be extended to an arbitrary line segment, i.e., without restricting it to be horizontal. We observed that we can extend this result to any dimension under both the $\mathbb{L}_2$ and $\mathbb{L}_\infty$ norms.

\begin{observation}\label{fac_on_line_convex}
Consider a point $p$ and a line $L$ in $\mathbb{R}^d$. Now for all $q_1,q_2 \in L$ if $q_1 \geq_L q_2 \geq_L p^*$ or $q_1 \leq_L q_2 \leq_L p^*$ then $||p-q_1||\geq ||p-q_2||$ under both $\mathbb{L}_2$ and $\mathbb{L}_\infty$ norm.
\end{observation}

\begin{observation}
 \label{fact_backpair_equality1}    Consider two backpair points $(p,q)$ for a line $L$ and optimal point for this backpair be $r$. If $r=_{L}p^*$, $|p-r| \geq |q-r|$; If $r=_{L}q^*$, $|q-r| \geq |p-r|$; otherwise $|p-r| = |q-r|$.
\end{observation}

\begin{lemma}\label{lm_berg_nec}
    $D_F(P,\lab)\geq \df{P}{\lab}$
\end{lemma}
\sloppypar
\begin{proof}
    By definition $D_F(P,\lab)\geq ||p_0-a||,\;||p_n-b||,\; \delta_{\overrightarrow{H}}(P,\lab)$. Now suppose $D_F(P,\lab)<B(P,L)$; so there exists some $p_i,p_j \in P$ such that $D_F(P,\lab)<B_{(p_i,p_j)}(L)$ $=d$ (say) with $i<j, \: p_i^*>_L p_j^*$. Let $q$ be the optimal point on $L$ for this backpair. From \cref{fac_opt_midl} $p_i^*\geq q \geq p_j^*$. If the equality satisfies in either side, $||p_i-p_i^*||$ or $||p_j-p_j^*||$ is equal to $d$ which implies $\delta_{\overrightarrow{H}}(P,\lab)\ge d$ contradicting our assumption $d>D_F(P,\lab)\geq \delta_{\overrightarrow{H}}(P,\lab)$.
    So, $p_i^*>q >p_j^*$ and from \cref{fact_backpair_equality1},
     $||p_i-q||=||p_j-q||=d$. Now, for all $q'>_L q$, $||p_j-q'||\ge d$ and for all $q''\le_L q$, $||p_i-q''||\ge  d$ according to \cref{fac_on_line_convex}.
     
    Let $\alpha, \beta$ be the mapping function corresponding to $P$ and $\lab$ respectively as mentioned in \cref{eq_frd_alt_gadou}. Now as, $D_F(P,\lab)<d$, if $p_i=P(\alpha(t))$ then $\lab(\beta(t)) >_L q$ otherwise $||P(\alpha(t))-\lab(\beta(t))||\geq d$ and $||D_F(P,\lab)||\geq d$. Let $p_j=P(\alpha(t'))$. So, $t'\geq t$ as $j>i$ and $\lab(\beta(t'))\geq_L \lab(\beta(t))>_L q$. But this leads to a contradiction as $||p_j-\lab(\beta(t'))||\ge_L ||p_j-q||=d$. So, $D_F(P,\lab)\ge d$.
\end{proof}

\begin{lemma}\label{lm_berg_suf}
    $D_F(P,\lab)\leq \df{P}{\lab}$
\end{lemma}

\begin{proof}
Say, $\df{P}{\lab}=d$. Alt and Godau~\cite{AltG95} and Authors in ~\cite{data_structure_frd} showed that if the following satisfies: $|p_0-a|,|p_0-b| \le d$ and for any $i<j$ with $p_i,p_j \in P$, there exist $a\le_L q \le_L q' \le_L b$ such that $|p_i-q|, |p_j-q'|\le d$, we can say $D_F(P,\lab)\leq d$.
 
As, $\df{P}{\lab}=d$, it comes straightforward that $|p_0-a|,|p_0-b| \le d$. Now for any $i<j$ where $p_i,p_j \in P$, $|p_i-p_i^*|, |p_j-p_j^*|\le d$. So, if $p_j^*\ge_L p_i^*$, we have $q=p_i^*$ and $q'=p_j^*$. When $p_j^*<_L p_i^*$, consider the points $p_i^e$ closest to $a$ and $p_j^e$ closest to $b$ for which $|p_i-p_i^e|\le d$ and $|p_j-p_j^e|\le d$.

If $p_j^e<_L p_i^e$, $||p_i-r||>d$ and $||p_j-r||>d$ for all points $r$ in the interval $[a,p_j^e)$ and $(p_i^e,b]$ respectively. Consequently there does not exist a point $r \in [a,b]$ such that $\max(||p_i-r||,||p_j-r||\le d)$ which gives rise to a contradiction as $\df{P}{\lab}=d \implies B(P,L)\le d$. So, $b\ge_Lp_j^e\ge_L p_i^e\ge_L a$ and the criterion stated above is satisfied.
\end{proof}

\cref{lm_berg_nec} and \cref{lm_berg_suf} constitute \cref{thm_frd}.

\begin{theorem}\label{thm_frd}
    For a curve $P$ and a line segment $\lab \in L$ in $\mathbb{R}^d$ the \frd between the curve and the line segment is $\delta_F(P,\lab)$ under both $\mathbb{L}_2$ and $\mathbb{L}_\infty$ norm where,
    $$\df{P}{\lab} = \max\{ ||p(0) - a|| , ||p(n) - b||, \dhl{P,\lab},\dhr{P,\lab},\\ \dhm{P,L}, B(P,L) \}$$
\end{theorem}

In this paper, we will deal with multiple polygonal curves. Let $\pi =\{\pi_1,\pi_2,...\}$ be a set of  polygonal curves, we will denote $i^{th}$ vertex of $k^{th}$ curve by $\pi_k(i)$. We define \frd between a line segment $\ell(a,b)$ and $\pi$ as: $$\delta_F(\pi,\lab)=\max_{\substack{i}}\delta_F(\pi_i,\lab)$$

We will only write the notation of the line segment (or the line where it belongs to) in the argument of the distance functions, where the polygonal curve is known from the context.

\section{Introduction to some distance functions}
We discussed some distance functions in \cref{sec_prelim}; in this section, we introduce a new type of distance function that is essential for computing the $(1,2)$-center of a set of polygonal curves.

\begin{definition}[Start-Backpair]
Two vertices from the set of the curves which include the start vertex of some curve satisfying the condition $\pi_k(0)^*>_L \pi_l(i)^*$ is referred to as \emph{start-backpair} w.r.t $L$. If both vertices are the start point, the one whose closest point on $L$ appears later while traversing $L$ in its direction is considered as the start vertex of the start-backpair.
\end{definition}

\begin{figure}[H]
    \centering
    \includegraphics[width=0.8\linewidth]{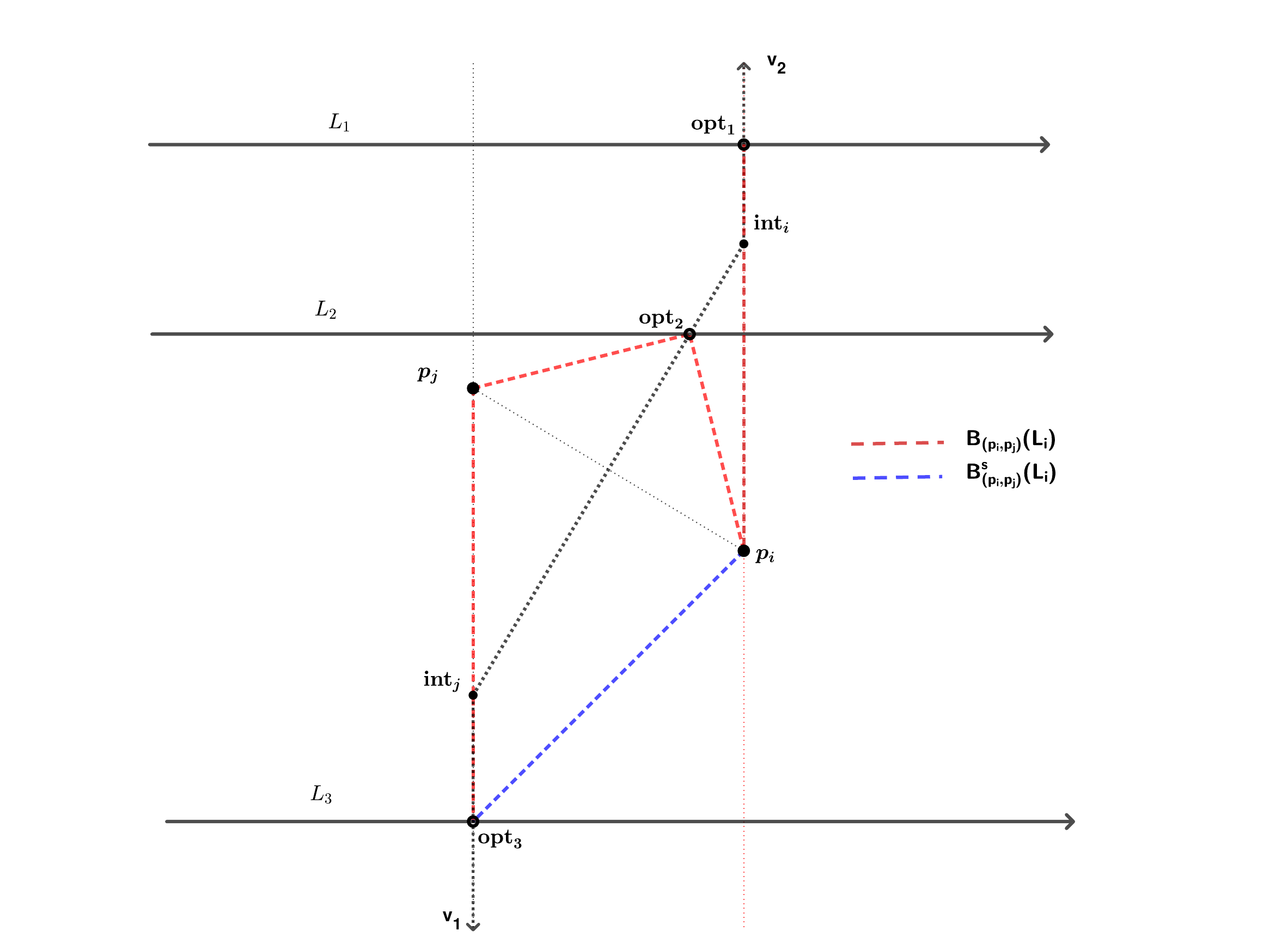}
    \caption{Intersection of line with path $\text{v}_1-\text{int}_j-\text{int}_j-\text{v}_2$ is the optimal point for $(p_i,p_j)$. Considering $p_i$ as the start point, for $L_1,L_2$ backpair/ start backpair distance is same, for $L_3$ backpair distance is $|p_j-\text{opt}_3|$ but start backpair distance is $|p_i-\text{opt}_3|$}
    \label{fig_optimal_bpoint}
\end{figure}

Although the criteria for the choice of vertices for backpair and start-backpair differ, the geometric condition is similar. The optimal point for a start-backpair and a backpair is also determined in the same way, but the start-backpair distance is defined differently.

\begin{definition}
    The distance of the start-vertex of the start-backpair from the optimal point for this backpair is considered as \emph{the start-backpair distance} even if the distance from the other point is higher. \[B^s_{(\pi_k(0),\pi_l(i))}(L)=\min (\min_{\substack{p\in L}}max(||\pi_k(0)-p||,||\pi_l(i)-p||),|\pi_k(0)-\pi_l(i)^*| )\]
\end{definition}

The backpair distance defined in \cref{sec_prelim} was first introduced in \cite{data_structure_frd}. In this section, we discuss the backpair distance along with the newly defined start backpair distance. In \cref{fig_optimal_bpoint}, the optimal point for a backpair and a start-backpair, along with their respective distances under $\mathbb{L_2}$ norm, is illustrated for a pair of vertices $(p_i, p_j)$, where $p_i$ is taken as the start point of the start-backpair. Consider the lines of orthogonal projection of $p_i,p_j$ on some line $L_k$. The intersection points of the perpendicular bisector of the line segment $\overline{p_i,p_j}$ with these lines are $\text{int}_i$ and $\text{int}_j$ respectively. If the bisector intersects $L_k$ between $\text{int}_i$ and $\text{int}_j$, the intersection point is the optimal point for $p_i,p_j$; the intersection point of $L_k$ with one of the orthogonal projection lines is the optimal point. Backpair distance is the maximum of the distance of the optimal point from the vertices in concern, but the start backpair distance is bounded by the distance of the optimal point to the start-vertex. 

\begin{definition}
Two vertices from the set of the curves which include the end vertex of some curve satisfying the condition $\pi_k(n) ^*$ < $_L \pi_l(i)^*$ are referred to as \emph{end-backpair} w.r.t $L$. If both vertices are endpoints, the one whose closest point on $L$ appears earlier while traversing $L$ in its direction is considered as the end vertex of the end-backpair. The end-backpair distance is defined as: \[B^e_{(\pi_k(n),\pi_l(i))}(L)=\min (\min_{\substack{p\in L}}max(||\pi_k(n)-p||,||\pi_l(i)-p||),|\pi_k(n)-\pi_l(i)^*| )\]
\end{definition}

The start and end-backpair distance of a polygonal curve from a set of polygonal curves with a line are defined in the following way-
$$B^s(\pi_k,L)=\max_{\substack{\forall\; i,\; l\\\pi_l(i)^*<_L\pi_k(0)^*}}B^s_{\pi_k(0),\pi_l(i)}(L)$$
$$B^e(\pi_k,L)=\max_{\substack{\forall\; i,\; l\\\pi_l(i)^*>_L\pi_k(n)^*}}B^e_{\pi_k(n),\pi_l(i)}(L)$$
Figure \ref{back_dist} illustrates the behavior of the backpair and start-backpair distance functions for curves in $\mathbb{R}^2$ centered at a fixed slope.

\begin{definition}
    Consider a backpair or start/end-backpair $(p_i,p_j)$ and any positive value $s$ greater than or equal to its backpair or start/end-backpair distance. The optimal point for this backpair is $q$. Let $x\le_Lq$ and $y\ge_L q$ be two such point so that for all $t<_L  x$, $|p_i-t|> s$ and for all $t >_L  y$, $|p_j-t|> s$ but $|p_i-x|=|p_j-y|= s$. We call this segment $\overline{xy}$ the $s$-optimal segment for this backpair or start/end-backpair on $L$.
\end{definition}

\begin{figure}[H]
    \centering    
    \includegraphics[width=0.8\linewidth]{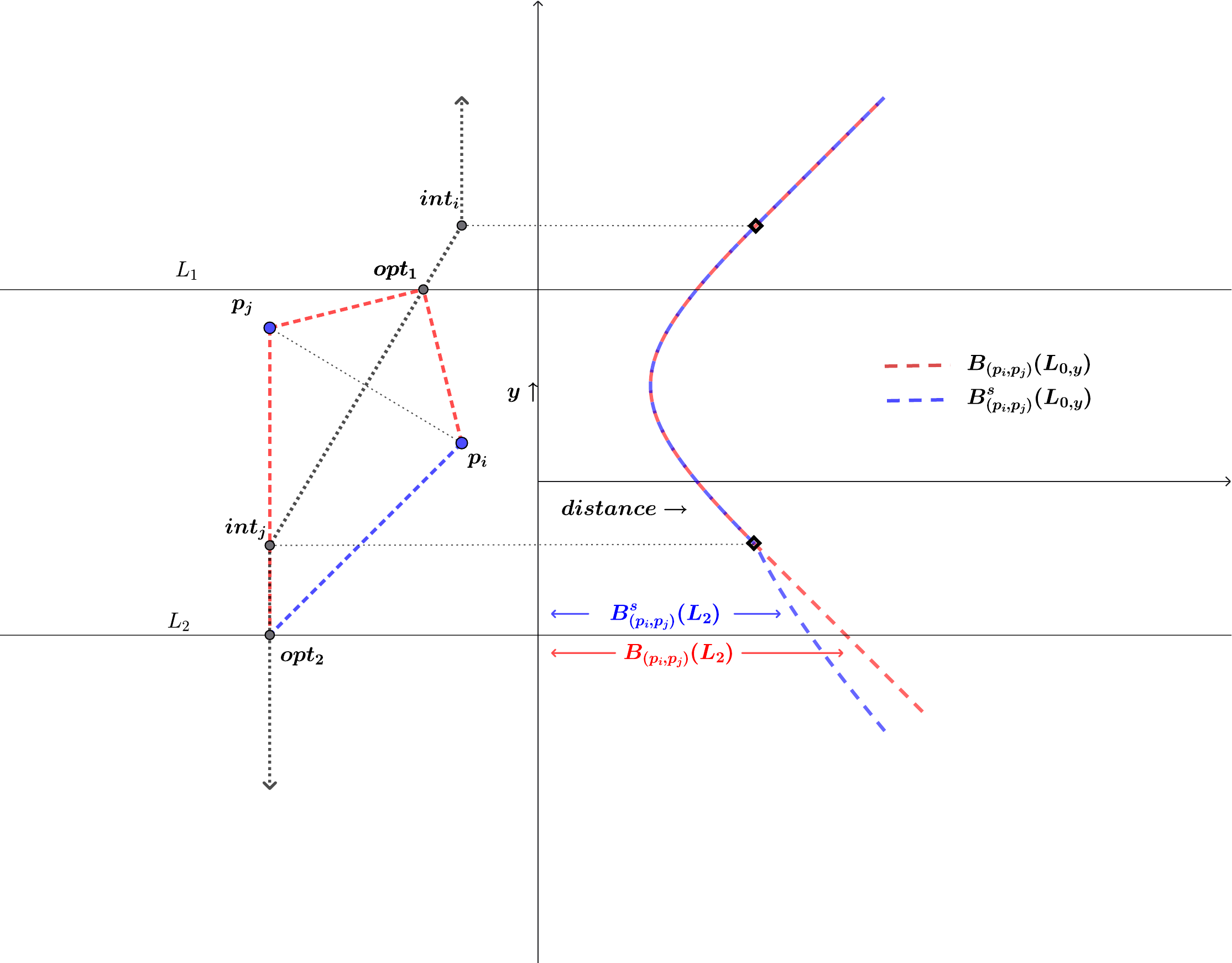}
    \caption{Backpair/ Start backpair distance function considering $p_i$ as start point}
    \label{back_dist}
\end{figure}

From \cref{fac_opt_midl} and \cref{fac_on_line_convex}, it follows that:
\begin{observation}
    The $s$-optimal segment for any start/end-backpair is continuous.
\end{observation}

\section{Upper Envelope of Functions}\label{sec_env}
Given a collection of functions $f_1(x_1,x_2,\ldots), f_2(x_1,x_2,\ldots),\ldots$, their upper envelope is
\[
  \U{x_1,x_2,\ldots} = \max_i f_i(x_1,x_2,\ldots).
\]

Consider a set $\pi=\{\pi_1,\pi_2,\ldots,\pi_r\}$ of polygonal curves. We construct the following upper envelopes:
\begin{align*}
  \mathcal{U}_g(x_1,x_2,\ldots) &= \max_i \max\bigl( B(\pi_i,L_{x_1,x_2,\ldots}),\,
                                    \dhm{\pi_i,L_{x_1,x_2,\ldots}} \bigr),\\
  \mathcal{U}_s(x_1,x_2,\ldots) &= \max_i B^s(\pi_i,L_{x_1,x_2,\ldots}),\\
  \mathcal{U}_e(x_1,x_2,\ldots) &= \max_i B^e(\pi_i,L_{x_1,x_2,\ldots}).
\end{align*}

We refer to these simply as $\mathcal{U}_g$, $\mathcal{U}_s$, $\mathcal{U}_e$, and write $\mathcal{U}=\max(\mathcal{U}_g,\mathcal{U}_s,\mathcal{U}_e)$. The functions $B^e_{(\pi_k(n),\pi_l(i))}(L), B^s_{(\pi_k(0),\pi_l(i))}(L), B_{(p_i,p_j}(L)$ are partially defined as the vertices under consideration remains backpair for a certain range of parameter values of $L$.
For $r$ curves with maximum complexity $n$, we have $O(nr^2)$ start and end backpair distance functions contributing to the upper envelopes $\mathcal{U}_s$ and $\mathcal{U}_e$ as we need to consider at most all possible pairs among $nr$ vertices and $r$ start or end vertices. For backpair function at most we need to consider every pair of vertices from each curve so $O(n^2r)$ backpair functions contribute to the upper envelope $\mathcal{U}_g$ and $\dhm{\pi_i, L}$ computes at most $n$ distance functions $|\pi_i(j)-\pi_i(j)^*|$ for all the $r$ curves so no of distance functions is bounded by $O(nr)$ for this. So, the number of underlying distance functions to compute the upper envelopes $\mathcal{U}_s\text{ or }\mathcal{U}_e \text{ and } \mathcal{U}$ is $ O (nr^2)\text{ and }O(n^2r+nr^2)$.

All the distance functions considered in the upper envelopes are partially or totally defined $2(d-1)$-variate functions of constant algebraic degree and constant complexity in both $L_2$ and $L_\infty$ norms. When the curves are restricted to be in $\mathbb{R}^2$, all these functions are surface patches in $3$-dimensional space. Results in \cite{envelope_bound} about construction of upper envelope of such functions reveals that we can construct the upper envelope $\mathcal{U}_s$, $\mathcal{U}_e$, $\mathcal{U}$, in $O\bigl((n^2r+nr^2)^{2+\epsilon}\bigr)$ time for any $\epsilon>0$ and the constants for exact time depend on $\epsilon$. \footnote{In \cite{buchin_eff_frech} the results in \cite{envelope_bound} were also used to obtain the runtime for construction of the upper envelope $ \max_i\bigl( B(\pi_i, L_{x_1,x_2})$ under $\mathbb{L}_2$ norm.}

In \cite{env4d} it has been shown that the expected time to construct the upper envelope of surfaces in $d$-space is $N^{d+\epsilon}$ for any $\epsilon>0$, where $N$ is the number of functions. So we can compute the upper envelope $\mathcal{U}$ in $O\bigl((n^2r+nr^2)^{2(d-1)+\epsilon}\bigr)$ expected time for curves in dimension $d$.

\subsection{Upper envelope in a restricted case}
For curves in $\mathbb{R}^2$, when the center is restricted to lie horizontally,
each function contributing to the upper envelope $\mathcal{U}$ or $\mathcal{U}_s, \mathcal{U}_e$ namely
$||\pi_i(k)-\pi_i(k)^*(L_{0,y})||$, $B_{(\pi_i(j),\pi_i(j'))}(L_{0,y})$,
$B^s_{(\pi_i(0),\pi_k(j))}(L_{0,y})$, and
$B^e_{(\pi_i(n),\pi_k(j))}(L_{0,y})$-is totally defined and unimodal convex in $y$.
This structure lets us adapt the prune-and-search paradigm of
Megiddo~\cite{Megiddo} to compute $\min_y \mathcal{U}(y)$ in linear time.

\paragraph{Decomposition:}
Since the maximum of convex functions is convex, $\mathcal{U}$ itself is convex, and hence has a unique global minimum at $y^{*}$. Each unimodal convex function contributing to $\mathcal{U}$ can be split, at some $y$ where the minimum of that function is attained, into a non-increasing piece and a non-decreasing piece. If $t$ functions contribute to $\mathcal{U}$, this yields a set of $2t$ partially defined monotonic functions $F = \{F_1, F_2, \ldots, F_{2t}\}$, of which exactly $t$ are non-decreasing and $t$ are non-increasing.

\begin{lemma}[Decision procedure]\label{lem_decision}
For any test value $y_m$, we can determine in $O(|F|)$ time whether
$y_m = y^{*}$, $y_m < y^{*}$, or $y_m > y^{*}$.
\end{lemma}

\begin{proof}
Evaluate every $F_i \in F$ at $y_m$ and let $F_{m_1}, F_{m_2}, \ldots$ be the functions attaining the maximum value there, i.e.\ the pieces that are active at $y_m$ (this determines $\mathcal{U}(y_m)$ and takes $O(|F|)$ time).
\begin{itemize}
  \item If every active piece $F_{m_j}$ is non-increasing at $y_m$ (or has
  $y_m$ as its top boundary point), then $\mathcal{U}$ is
  decreasing at $y_m$ and $y^{*} > y_m$.
  \item If every active piece $F_{m_i}$ is non-decreasing at $y_m$ (or has
  $y_m$ as its bottom boundary point), then $\mathcal{U}$ is 
  increasing at $y_m$ and $y^{*} < y_m$..
  \item Otherwise, some active piece is non-increasing, and another is
  non-decreasing at $y_m$, so by convexity $y_m = y^{*}$.
\end{itemize}
\end{proof}

\begin{figure}[H]
    \centering
\includegraphics[width=0.9\linewidth]{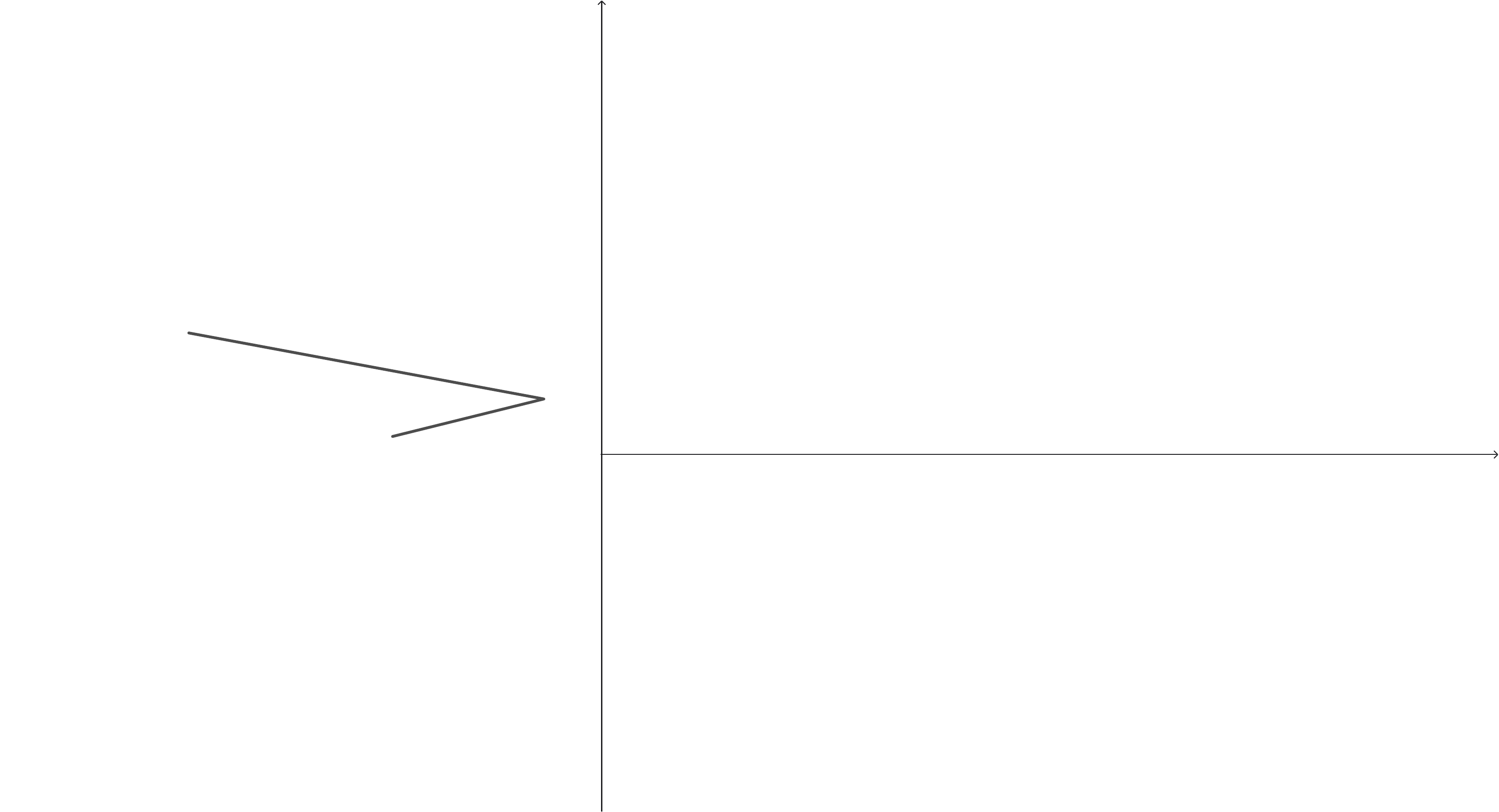}
    \caption{Upper envelope is contributed by the start backpair function and projection distance functions}
    \label{Up_env}
\end{figure}

\paragraph{Pruning step:}
Arbitrarily pair two functions of $F$ to produce $t$ such pairs $p_1, p_2, \ldots, p_t$ (There may be one function left unpaired when $F$ gets reduced). Our goal in each step is to identify some functions that do not contribute to the upper envelope at its minimum. For each pair $p_i$, let $y_i$ be the $y$-coordinate of the intersection of its two functions; if the pair never intersects, we immediately remove the functions with the lower value. As the functions are monotone in $y$, one function of a pair $p_i$ dominates for $y < y_i$ and the other one dominates for $y > y_i$. Let's name these functions:$F^-_i, F^+_i$ respectively.

Now compute the median $y_M$ of $\{y_1, y_2, \ldots, y_t\}$ in $O(t)$ or $O(|F|)$ time. In $O(|F|)$ time we can determine whether
$y_M = y^{*}$, $y_M < y^{*}$, or $y_M > y^{*}$ (\cref{lem_decision}).

If $y^{*} = y_M$, we are done. If $y^{*} > y_M$, for every pair $p_i$ with $y_i \le y_M$, $F_i^+$ dominates $F_i^-$ throughout the region $y \geq y_M$, so $F_i^-$ cannot contribute to $\mathcal{U}$ at $y^{*}$ and may be discarded. As $y_M$ is the median of $y_i$'s, at least half of the $t$ pairs satisfy $y_i \leq y_M$, so this eliminates one function from each such $\frac{t}{2}$ pairs. So we can discard at least $ O(\frac{|F|}{4})$ functions taking the unpaired function and pairs where functions do not intersect into consideration too. 

Analogously, if $y^{*} < y_M$, $O(|F|/4)$ functions are eliminated.

\paragraph{Running time:}
By repeating the pruning procedure, we can determine $y^*$, where the upper envelope attains its minimum, as we go on discarding the functions that do not contribute to the upper envelope at $y^*$. Each pruning round removes $ O(\frac{|F|}{4})$ functions in $O(|F|)$ time. The recurrence relation function for the runtime complexity of computing $y^*$ is:
$$T(|F|) = T(O(\frac{3}{4}|F|)) + O(|F|) = O(|F|)$$

So, it follows that for curves in $\mathbb{R}^2$ and while the center is restricted to be horizontal, we can find the minimum value of the upper envelope along with the corresponding parameter values in $O(n^2r+nr^2)$ time.

\section{Exact Algorithm for the \texorpdfstring{$(1,2)$}{(1,2)}-Center Problem}
Consider the upper envelopes:$\mathcal{U}_s,\mathcal{U}_e$ and $\mathcal{U}$. We find $(x_1^0,x_2^0,\ldots x_{2(d-1)}^0)$ for which the minimum value of $\mathcal{U}(x_1,x_2,\ldots x_{2(d-1)})$ is attained. We find the start and end backpair contributing to the $\mathcal{U}_s$ and $\mathcal{U}_e$ at $(x_1^0,x_2^0,\ldots x_{2(d-1)}^0)$. The optimal point for that start and end backpair be $a,b$ on $L_{x_1^0,x_2^0,\ldots x_{2(d-1)}^0}$  respectively. The $(1,2)$-center for the family of polygonal curves is $\overline{ab}$.
\subsection[(1,2)-center algorithm]{\texorpdfstring{$\mathbf{(1,2)}$}{(1,2)}-center algorithm}

\begin{algorithm}[H]
\DontPrintSemicolon
\caption{\textsc{Optimal-$(1,2)$-Center}}
\label{alg:1-2-center}
\KwIn{A set $\pi = \{\pi_1, \pi_2, \ldots, \pi_r\}$ of polygonal curves in $\mathbb{R}^d$.}
\KwOut{A line segment $\overline{ab}$ realizing an optimal $(1,2)$-center, and its value $D = \max_i \delta_F(\pi_i, \overline{ab})$.}
\BlankLine
Construct the upper envelopes $\mathcal{U}_s$, $\mathcal{U}_e$, and $\mathcal{U} = \max(\mathcal{U}_g, \mathcal{U}_s, \mathcal{U}_e)$\;

Compute the parameter values $(x_1^0, x_2^0, \ldots, x_{2(d-1)}^0)$ for which $\mathcal{U}$ attains minimum value, and set $D \gets \mathcal{U}(x_1^0, x_2^0, \ldots, x_d^0)$\;
\label{line:minimize}
Let $L \gets L_{x_1^0, x_2^0, \ldots, x_{2(d-1)}^0}$\;
\BlankLine
\tcp{Recover the start endpoint from the start backpair}
Find the $D$-optimal segment for all start backpair $\bigl(\pi_k(0), \pi_l(i)\bigr)$ w.r.t $L_{x_1^0, x_2^0, \ldots, x_{2(d-1)}^0}$\;
\label{line:start}
$a \gets$ arbitrary point from intersection of all $D$-optimal segment considered in  step 4.\;
\BlankLine
\tcp{Recover the end endpoint from the end backpair}
Find the $D$-optimal segment for all end backpair $\bigl(\pi_k(n), \pi_l(j)\bigr)$ w.r.t $L_{x_1^0, x_2^0, \ldots, x_{2(d-1)}^0}$\;
\label{line:end}
$b \gets$ arbitrary point from intersection of all $D$-optimal segment considered in step 6.\;
\BlankLine
\Return the segment $\overline{ab} \in L$ and the value $D$\;
\end{algorithm}

\paragraph{Time complexity:}
In \cref{sec_env} we have discussed the time to construct the upper envelope $\mathcal{U}$. Step 1 of the algorithm is about the construction of the upper envelope. The time for executing other steps - finding the minimum value of the upper envelope or identifying the pair contributing to the envelope at some point- is dominated by the construction time of the upper envelope. So, the algorithm will produce the $(1,2)$-center in $O\bigl((n^2r+nr^2)^{2+\epsilon}\bigr)$ for any $\epsilon>0$ for curves in $\mathbb{R}^2$. And, in $O\bigl((n^2r+nr^2)^{2(d-1)+\epsilon}\bigr)$ expected time, we can construct the upper envelope for curves in dimension $d$.

Further, if the $(1,2)$-center is restricted to be horizontal and input curves lie on a plane, it can be found in $O(n^2r+nr^2)$ time as discussed in \cref{sec_env}.

\subsection{Correctness}
Before we prove the correctness of the algorithm, we will state a modified version of \cref{fact_backpair_equality1}.

\begin{observation}
 \label{fact_backpair_equality}    Consider two backpair points $(p,q)$ for a line $L$ and optimal point for this backpair be $r$. If $r=_{ab}p^*$, $|p-r| \geq |q-r|$; If $r=_{ab}q^*$, $|q-r| \geq |p-r|$; otherwise $|p-r| = |q-r|$. Also note: if it is a start(end)-backpair with start point $p$, the output of the function is $|p-r|$.
\end{observation}

\begin{lemma}
\label{lm_nonemptyD}
    The intersection of the $D$-optimal segment of all start/end-backpair is non-empty.
\end{lemma}

\begin{proof}
    Upon adding a single start-backpair, the statement is trivially true.
    Now consider $\overline{pq}$ to be the segment of intersection of the $D$-optimal segment after adding some start-backpairs. Certainly there exist a start vertex $\pi_k(0)$ such that $|\pi_k(0)-t|>D$ for all $t<_L p$ and a vertex $\pi_l(i)$ such that $|\pi_l(i)-t|>D$ for all $t>_L q$ where $t\in L$.
    
    On adding the next one, let the $D$-optimal segment for this start-backpair $\overline{mn}$ not intersect with $\overline{pq}$. 
    \begin{itemize}
        \item$m>_Lq:$  In this case, the start vertex of the newly added backpair $\pi_r(0)$ satisfies $|\pi_r(0)-t|>D$ for all $t<_L m$. Again there exist a vertex $\pi_l(i)$ such that $|\pi_l(i)-t|>D$ for all $t>_L q$. Consequently $B^s_{\pi_r(0),\pi_l(i)}(L)>D$ - a contradiction.
        \item $n<_Lp:$ In this case, the vertex other than the start vertex of the newly added start-backpair $\pi_s(j)$ satisfies $|\pi_s(j)-t|>D$ for all $t>_L n$. Again there exist a start point $\pi_k(0)$ such that $|\pi_k(0)-t|>D$ for all $t<_L p$. Consequently $B^s_{\pi_k(0),\pi_s(j)}(L)>D$ - a contradiction.
    \end{itemize}

So, $\overline{mn}$ and $\overline{pq}$ must intersect. Analogously, we can prove this for end backpairs.
\end{proof}

\begin{lemma}
\label{lm_left1}
For all vertex $p\in \{\pi_l(i): \pi_l(i)<_L a\}\cup \{\pi_k(0)\: \forall\: k\}$,  $|p-a|\le D$.    
\end{lemma}

\begin{proof}
If there does not exist any start-backpair $i.e.$ $\pi_k(0)^*$ coincides for all $k$ and for every non-start vertices satisfy $\pi_l(i)^*\ge_L a$ then , $|\pi_k(0)-\pi_k(0)^*|\le D$ as $\dhm{\pi_k,L}$ is contributing to upper envelope. Otherwise, all start vertices and any other vertex $\pi_l(i)$ such that $\pi_l(i)<_L a$ are part of some start backpair. Let $D$-optimal segment of one start-backpair where $\pi_l(i)$ is a part of be $\overline{xy}$. Certainly $\pi_l(i)^*\le_L a \le_L y$. So by \cref{fac_on_line_convex} $|\pi_l(i)-a| \le |\pi_l(i)-y|=D$. Now $\pi_k(0)^*>_L a$ be some start vertex and $D$-optimal segment for some start-backpair where $\pi_k(0)$ is the start vertex be  $\overline{x'y'}$. Now $x'\le_La<_L \pi_k(0)^* $. Consequently $D=|\pi_k(0)-x'|\ge |\pi_k(0)-a|$.
\end{proof}

\begin{lemma}
\label{lm_right1}
For all vertex $p\in \{\pi_l(i): \pi_l(i)>_L b\}\cup \{\pi_k(n)\: \forall\: k\}$,  $|p-b|\le D$.    
\end{lemma}

\begin{proof}
    The proof is symmetric to that of \cref{lm_left1}, hence
    it is omitted.
\end{proof}

\begin{lemma}
\label{necessary}
Any line segment in the plane has \frd greater than or equal to $D$, the minimum value of the upper envelope of $\mathcal{U}$, with at least one of the curves.
\end{lemma}

\begin{proof}
Say, a line segment $a'b'$ on line $L'_{x_1',x_2',\ldots}$ has \frd less than $D$ (minimum value of upper envelope) with all the curves. Let $D' \ge D$ be the value of upper envelope at $(x_1',x_2',\ldots)$. Consider the particular distance function contributing to the upper envelope at $(x_1',x_2',\ldots)$. If it is under $\dhm{L'}$ or $B(L')$ and the particular distance function is $||\pi_p(i)-\pi_p(i)^*||$ or $B_{(\pi_p(i),\pi_p(j))}(L'_{x_1',x_2',\ldots})$, according to \cref{thm_frd},  $\delta_F(\pi_p,a'b') \geq D' \ge D$.

We will show that if it is due to $B^s_{(\pi_k(0),\pi_l(i))}(x_1',x_2', \ldots)$, we can claim $\max\bigl(\delta_F(\pi_k,a'b'),\\ \delta_F(\pi_l,a'b')\bigr) \geq D'$. Similarly if upper envelope at $(x_1',x_2',\ldots)$ is contributed by end backpair function  $B^e_{(\pi_k(n),\pi_l(j))}(L'_{x_1',x_2',\ldots})$, we can claim $\max\bigl(\delta_F(\pi_k,a'b'),\delta_F(\pi_l,a'b')\bigr) \geq D'$. Proof is presented for the case where the upper envelope is contributed by the start backpair function; the case of the end backpair function follows analogously.

Let the optimal point for the start backpair $\pi_k(0), \pi_l(i)$ contributing to $\mathcal{U}_s$ at $(x_1',x_2',\ldots)$ be $r$ on $L'_{x_1',x_2',\ldots}$. According to construction of start backpair function, $\pi_k(0)^*\ge_{L'} r \ge_{L'} \pi_l(i)^*$.

\begin{figure}[ht!]
  \centering
  \begin{subfigure}[b]{0.495\linewidth}
    \centering
    \includegraphics[width=\linewidth]{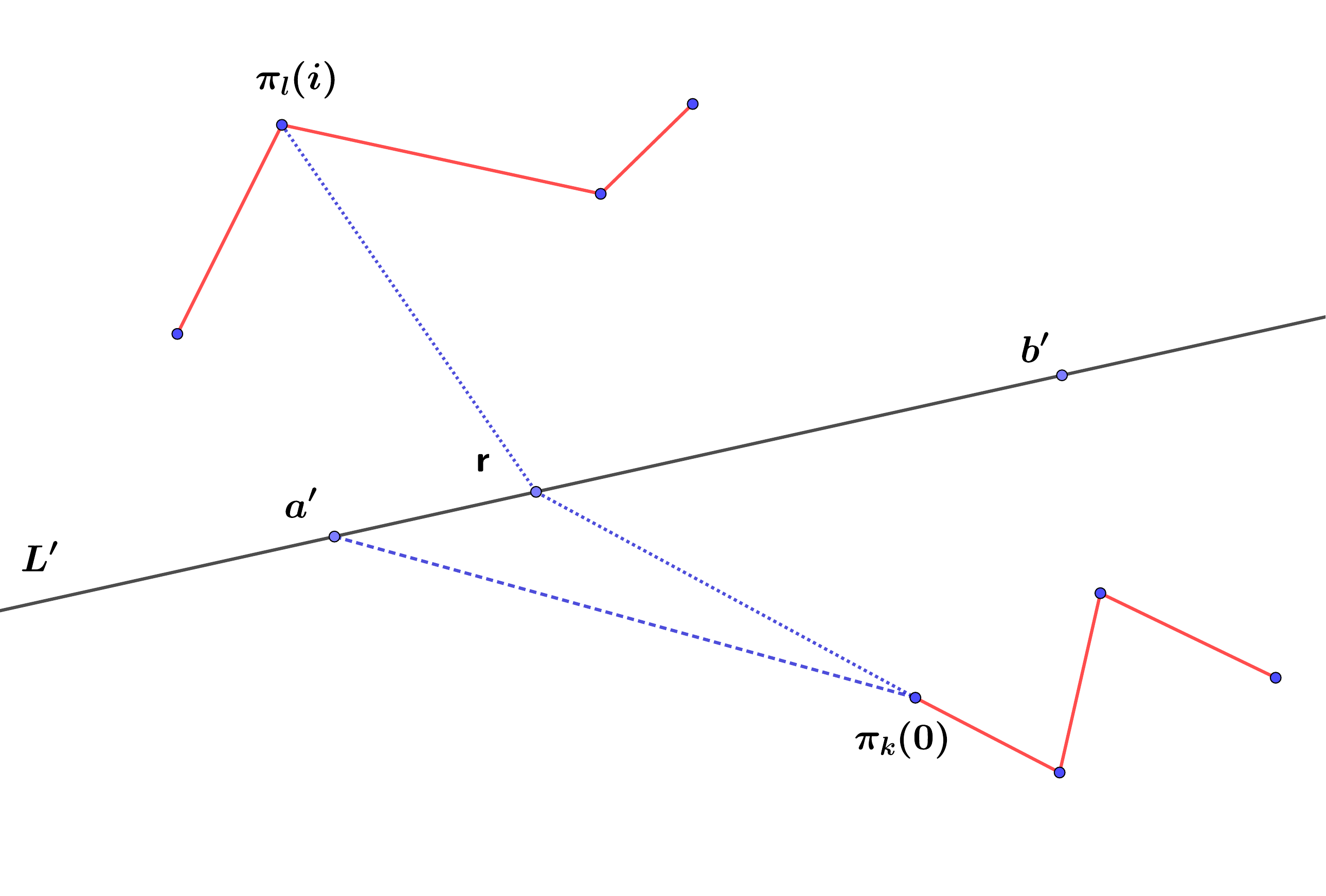}
    \caption{Case~1: $r \geq_{L'} a'$}
    \label{nec_start_c1}
  \end{subfigure}\hfill
  \begin{subfigure}[b]{0.495\linewidth}
    \centering
    \includegraphics[width=\linewidth]{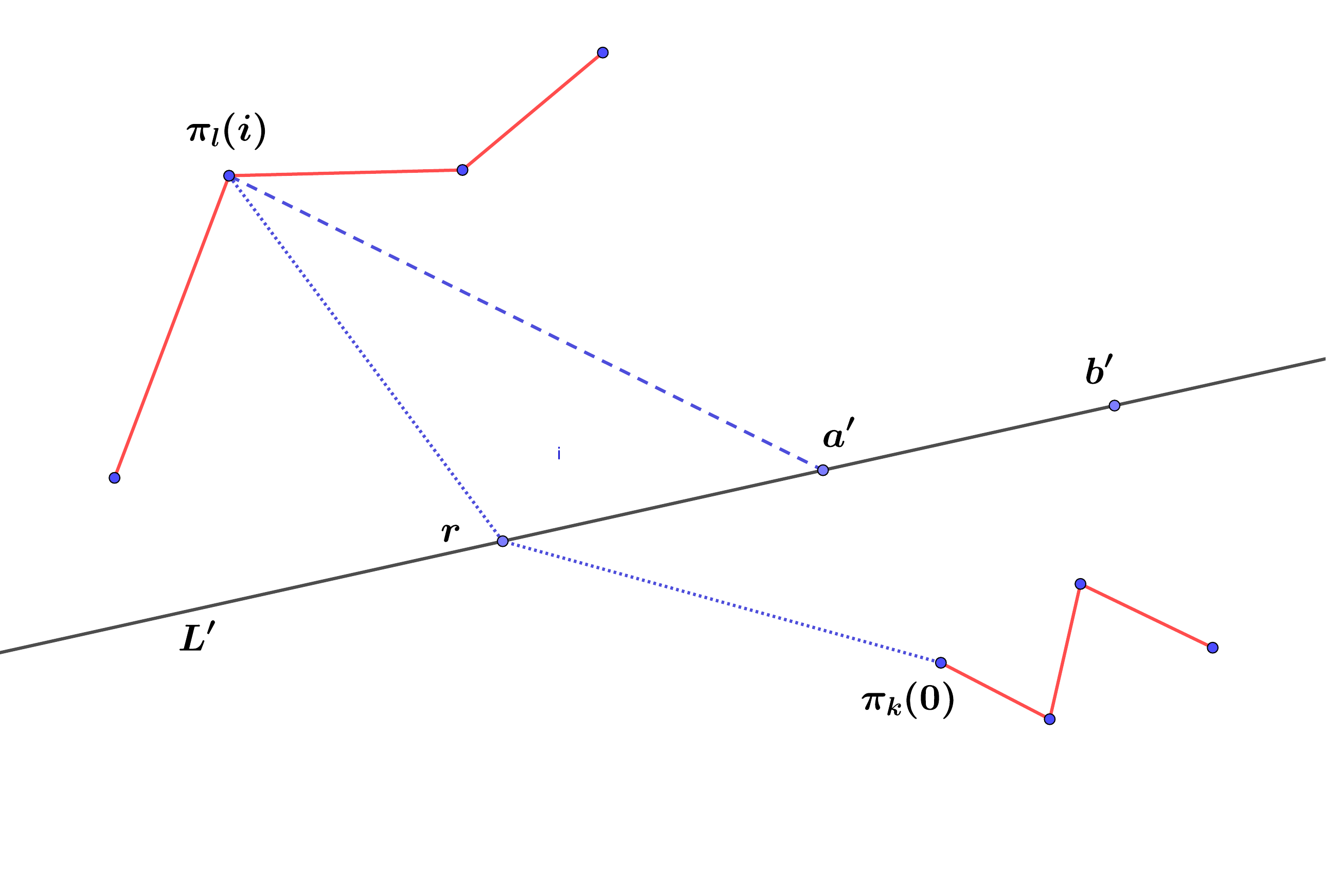}
    \caption{Case~2: $r <_{L'} a'$}
    \label{nec_start_c2}
  \end{subfigure}
  \caption{Illustration of \cref{necessary}.}
  \label{nec_start}
\end{figure}

\begin{description}[topsep=0pt, itemsep=0pt, parsep=0pt]
    \item[Case 1: $r\ge_{L'} a'$:]  As $a'\leq_{L'} r\leq_{L'} \pi_k(0)^*$, according to \cref{fac_on_line_convex}, $|\pi_k(0)-a'|\geq |\pi_k(0)-r|=D'\ge D$. So by \cref{thm_frd}, $\delta_F(\pi_k,\ell(a'b'))\geq D$. (Refer to figure \cref{nec_start_c1})
    \item[Case 2: $r\le_{L'} a'$:]  If $|\pi_k(0)-a'| \ge D$, then $\delta_F(\pi_k,\ell(a'b'))\geq D$. Suppose $|\pi_k(0)-a'| < D$ then $|\pi_l(i)-a'| \ge D$ otherwise $a'$ would be the optimal point of this start backpair. Again, $\pi_l(i)^*\le_{L'} r \le_{L'} a'$ according to construction. Every point $q \in a'b'$ satisfies the relation $q \ge_{L'} a'$. So according to \cref{fac_on_line_convex} for every such point $q$, $|\pi_l(i)-q|\ge |\pi_l(i)-a'| \ge D$. So by \cref{thm_frd} $\delta_F(\pi_l,a'b')\geq D$. (Refer to \cref{nec_start_c2})
\end{description}
\end{proof}

\begin{lemma}
\label{sufficient}
    The line segment $\overline{ab}$ determined by the algorithm has \frd less than or equal to $D$ with any of the curves.
\end{lemma}
\begin{proof}
    Consider a certain curve $\pi_m$. It follows directly from \cref{lm_left1,lm_right1}, $|\pi_m(0) - a|,|\pi_m(n) - b|, \dhl{\pi_m,\lab}, \dhr{\pi_m,\lab}\leq D$. It remains to show ${B(L)}\leq D$.\\
    Let $\pi_m(i), \pi_m(j)$ be a backpair with $i<j$. We need to prove there exist a point $p$ on $L_{x_1^0, x_2^0,\ldots}$ such that $b\geq_L p\geq_L a$ which satisfies the condition $\max(|\pi_m(i)-p|, |\pi_m(j)-p|)\le D$. If we show there exists a point (on $L_{x_1^0, x_2^0,\ldots}$) $p_s$ such that $p_s\geq_L a$ and $p_t$ such that $p_t\leq_L b$, which satisfies the condition, then we are done because all points in between $p_s$ and $p_t$ satisfy the condition. We will show that such a point $p_s$ exists. The existence of $p_t$ can be proved analogously and is therefore skipped here.
    
    \begin{figure}[H]
        \centering
        \includegraphics[width=0.8\linewidth]{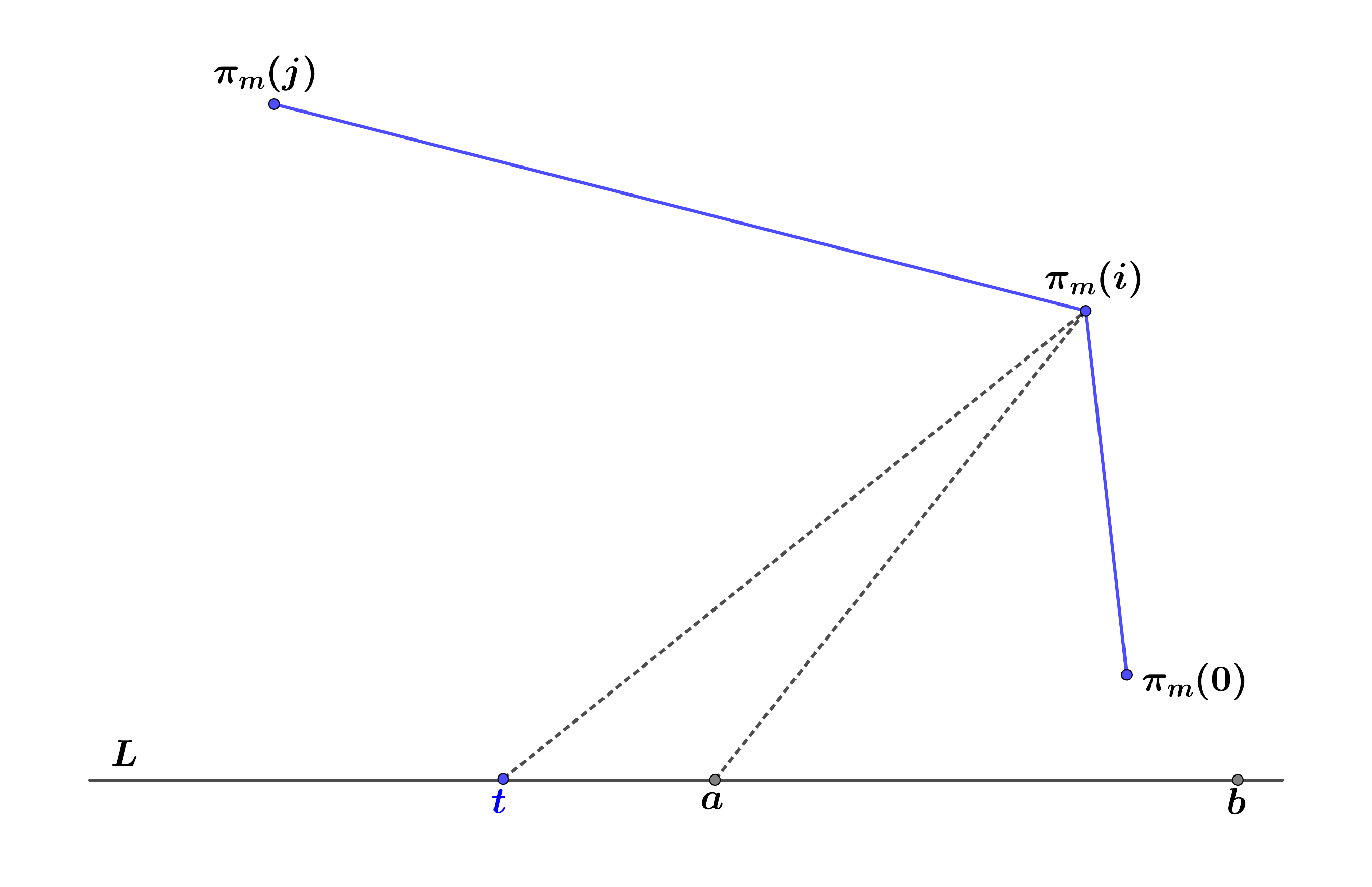}
        \caption{$B(\pi_m,L_{x_1^0, x_2^0,\ldots})\leq D$}
        \label{suf_back}
    \end{figure}
 
    If $a\leq_L \pi_m(j)^*$, the optimal point of this backpair $q$ satisfies $q>_La$ so we are done with $p_s:=q$. 
    Again if $a \geq_L \pi_m(i)^*$, both $|\pi_m(i)-a|, |\pi_m(j)-a| \le D$ as shown in \cref{lm_left1} and we have $p_s:=a$. So we will consider the case $\pi_m(i)^*\geq_L a\geq_L \pi_m(j)^*$. The backpair distance of the pair of points and the line $L$ is less than or equal to $D$. Let $t$ be the optimal point on $L$ for which $\max(|\pi_m(i)-t|, |\pi_m(j)-t|)$ is minimum (Refer to \cref{suf_back}). If $t \ge_L a$, we are done, since we can take $ p_s = t$. Let,  $t<_L a$. $|\pi_m(j)-a| \le D$ by \cref{lm_left1}. By \cref{fac_on_line_convex}, $|\pi_m(i)-a|\le |\pi_m(i)-t|\le D$. So, we have $p_s:=a$.
\end{proof}

From \cref{necessary} it follows that any $(1,2)$-center has at least distance $D$ from one of the curves, and \cref{sufficient} suffices that the $(1,2)$-center produced by the algorithm does not have \frd more than $D$, the minimum value of the upper envelope, with any of the curves. This leads to the following theorem- 
    \begin{theorem}
        The $(1,2)$-center algorithm produces an optimal $(1,2)$-center for curves. The runtime of the algorithm when $r,n$ are the number of curves and the maximum complexity of any curve, respectively-\\
        (a) $O\big((n^2r + nr^2)^{2(d-1)+\epsilon}\big)$ (expected) for curves in dimension $d$.\\
        (b) $O\big((n^2r + nr^2)^{2+\epsilon}\big)$ for curves in the plane.\\
        (c) $O(n^2r+nr^2)$ for the curves in the plane when the center is restricted to be horizontal.
    \end{theorem}

\section{Approximation Results}
\Cref{alg:1-2-center} computes an optimal $(1,2)$-center. We will now introduce methods for computing the center using different approximation schemes.

\subsection{\texorpdfstring{$1+\epsilon$}{1+epsilon}- factor approximation}
\label{sec_approx_reduction}
The exact algorithm finds the parameter values for which the upper envelope $\mathcal{U}$ of the distance
functions attains the minimum value and then finds two endpoints of the center on that line. Now we replace each distance function by a polynomial function that
approximates it within a fixed factor. Thus, we get the upper envelope that approximates the actual upper envelope within a fixed factor. We will introduce an approximation scheme for curves in $\mathbb{R}^2$, with distance measured under the $\mathbb{L}_2$ norm.

\subsubsection{Algebraic expression of Distance Functions}
Before presenting the algebraic expressions for the distance functions, we introduce the necessary notations. Intersection point of perpendicular bisector of two points $(p_i,p_j)$ with a line $L$ is denoted as $(p_{ij}^*)$.
The following results are obtained by direct computation-

$(p_i^*.x,p_i^*.y):= (\frac{p_i.x+m(p_i.y-c)}{1+m^2},\frac{mp_i.x+m^2p_i.y+c}{1+m^2})$ and the projection distance is $\frac{|p_i.y-mp_i.x-c|}{\sqrt{1+m^2}}$.

$(p_{ij}^*.x,p_{ij}^*.y):=(\frac{p_i.y^2+p_i.x^2+p_j.y^2+p_j.x^2}{2(m\delta_y+\delta_x)}-\frac{c\delta_y}{m\delta_y+\delta_x},\frac{m(p_i.y^2+p_i.x^2+p_j.y^2+p_j.x^2)}{2(m\delta_y+\delta_x)}+\frac{c\delta_x}{m\delta_y+\delta_x})$. 

Where $\delta_x=p_i.x-p_j.x, \delta_y=p_i.y-p_j.y$.

\begin{equation}
    \dhm{p_i,L}=|p_i.y-m\cdot p_i.x-c|\cdot(1+m^2)^{-\frac{1}{2}}
\end{equation}

Say, $B^*_{p_i,p_j}(L_{m,c})=||p_i-p_{ij}^*||, \quad \widetilde{B}_{p_i,p_j}(L_{m,c})=||p_i-p_j^*||$. Computation shows:

\begin{equation}
\begin{cases}
   \widetilde{B}_{p_i,p_j}(L_{m,c})=(1+m^2)^{-\frac{1}{2}}\cdot\bigl(m^2(p_j.y-\Delta_{p_i})+m\cdot\Delta_{p_j}+c-m\cdot c -\Delta_{p_i}\bigr)^{\frac{1}{2}}\\
    B^*_{p_i,p_j}(L_{m,c})=\frac{1}{\sqrt{2}}\cdot (m\cdot\delta_y+\delta_x)^{-\frac{1}{2}}\cdot (m\cdot v_y+s.c-v_x)^{-\frac{1}{2}} \text{ where,}
\end{cases}
\end{equation}

\begin{equation}
\begin{cases}
    \Delta_{p_i}=p_i.x+p_i.y\\[4pt]
    s=2(\delta_x-\delta_y)\\[4pt]
    v_y=p_i.y^2+p_i.x^2+p_j.y^2+p_j.x^2-2\delta_y\cdot p_i.y -\delta_y\cdot p_i.x\\[4pt]
    v_x=p_i.y^2+p_i.x^2+p_j.y^2+p_j.x^2-2\delta_x\cdot p_i.x -\delta_x\cdot p_i.y\\
\end{cases}
\end{equation}

are constants depending on $p_i,p_j$.

\begin{equation}
B_{(p_i,p_j)}(L)=
    \begin{cases}
        \sqrt{(\frac{p_i.x-p_j.x}{2})^2+(\frac{p_i.y-p_j.y}{2})^2} & \text{if} \;
            m = -\frac{\delta_x}{\delta_y}\\[16pt]
        B^*_{(p_i,p_j)}(L)
        & \begin{substack}{ \text{if} \; (x^*-p_i^*.x)(x^*-p_j^*.x)\leq 0\\ \text{and, } \; m \neq -\frac{\delta_x}{\delta_y}} \end{substack}\\[20pt]
        \dhm{p_i,L} & \begin{substack}{
            \text{if}\;m \neq -\frac{\delta_x}{\delta_y} \; \text{ and,}\\ (x^*-p_i^*.x)(x^*-p_j^*.x)> 0 \text{ and,}\\
            |x^*-p_j^*.x|>|x^*-p_i^*.x|}
        \end{substack}\\[20pt]
        \dhm{p_j,L} & \begin{substack}{
            \text{if}\; m \neq -\frac{\delta_x}{\delta_y} \; \text{ and,}\\ (x^*-p_i^*.x)(x^*-p_j^*.x)> 0 \text{ and,}\\ |x^*-p_i^*.x|>|x^*-p_j^*.x|}
        \end{substack}
    \end{cases}
\end{equation}

\begin{equation}
B^s_{(p_0,p_i)}(L)=
    \begin{cases}
        \sqrt{(\frac{p_i.x-p_0.x}{2})^2+(\frac{p_i.y-p_0.y}{2})^2} & \text{if} \;
            m = -\frac{\delta_x}{\delta_y}\\[16pt]
        B^*_{(p_i,p_0)}(L)
        & \begin{substack}{ \text{if} \; (x^*-p_i^*.x)(x^*-p_0^*.x)\leq 0\\ \text{and, } \; m \neq -\frac{\delta_x}{\delta_y}} \end{substack}\\[20pt]
        \widetilde{B}_{(p_0,p_i)}(L) & \begin{substack}{
            \text{if}\;m \neq -\frac{\delta_x}{\delta_y} \; \text{ and,}\\ (x^*-p_i^*.x)(x^*-p_0^*.x)> 0 \text{ and,}\\
            |x^*-p_0^*.x|>|x^*-p_i^*.x|}
        \end{substack}\\[20pt]
        \dhm{p_0,L} & \begin{substack}{
            \text{if}\; m \neq -\frac{\delta_x}{\delta_y} \; \text{ and,}\\[12pt] (x^*-p_i^*.x)(x^*-p_0^*.x)> 0 \text{ and,}\\[12pt] |x^*-p_i^*.x|>|x^*-p_0^*.x|}
        \end{substack}
    \end{cases}
\end{equation}

\begin{equation}
B^e_{(p_n,p_i)}(L)=
    \begin{cases}
        \sqrt{(\frac{p_i.x-p_n.x}{2})^2+(\frac{p_i.y-p_n.y}{2})^2} & \text{if} \;
            m = -\frac{\delta_x}{\delta_y}\\[16pt]
        B^*_{(p_n,p_i)}(L)
        & \begin{substack}{ \text{if} \; (x^*-p_i^*.x)(x^*-p_n^*.x)\leq 0\\ \text{and, } \; m \neq -\frac{\delta_x}{\delta_y}} \end{substack}\\[20pt]
        \widetilde{B}_{(p_n,p_i)}(L) & \begin{substack}{
            \text{if}\;m \neq -\frac{\delta_x}{\delta_y} \; \text{ and,}\\ (x^*-p_i^*.x)(x^*-p_n^*.x)> 0 \text{ and,}\\
            |x^*-p_n^*.x|>|x^*-p_i^*.x|}
        \end{substack}\\[20pt]
        \dhm{p_n,L} & \begin{substack}{
            \text{if}\; m \neq -\frac{\delta_x}{\delta_y} \; \text{ and,}\\ (x^*-p_i^*.x)(x^*-p_n^*.x)> 0 \text{ and,}\\ |x^*-p_i^*.x|>|x^*-p_n^*.x|}
        \end{substack}
    \end{cases}
\end{equation}

\subsubsection{Polynomial Approximation}
\begin{observation}
   The set of lines of form $I_1:y=mx+c$ and $I_2:x=my+c$ where $|m|\le 1, c\in \mathbb{R}$, constitutes the set of all lines in the plane. 
\end{observation}

\begin{proof}
    Consider an arbitrary line in the plane $L=ax+by+c=0$. It can be written as- $y=\frac{-a}{b}x-\frac{c}{b}$ or $x=\frac{-b}{a}y-\frac{c}{a}$. Now if $|\frac{a}{b}|<1$ it is a line within the set $I_1$, otherwise $|\frac{a}{b}|\ge 1$ and $L$ is a line of the set $I_2$. 
\end{proof}

Now, consider the smallest rectangle bounding all the vertices of the curves. Let the output of the $(1,2)-$ center algorithm $\lab \in L$. $L$ must pass through the bounding box; otherwise, one can bring it closer to the box and reduce the value of all distance functions, contradicting the optimality of $L$. It allows us to state-
\begin{observation}
    There exist constants $c_i\quad \text{for} \quad i\in\{1,2,\ldots 6\}$ such that the line $L$ containing the optimal $(1,2)$-Center can be presented by any of the following lines-
    $$y=mx+c\quad \text{where} \quad |m|\le1, c_1 \le c \le c_2$$
    $$x=my+c\quad \text{where} \quad |m|<1, c_3 \le c \le c_4$$
    $$y=x+c\quad \text{where} \quad  c_5 \le c \le c_6$$
\end{observation}

Since $m$ and $c$ range over bounded, closed intervals in each case, the
domain of each distance function $B^*(L)$, $\widetilde{B}(L)$, and
$\dhm{L}$ is compact. Let $f(L)$ be some function of the form $B^*(L)$,
$\widetilde{B}(L)$, or $\dhm{L}$. As each of these functions is continuous
on its respective domain, for any $\widetilde{\epsilon}>0$, there exists a
polynomial $p(m,c)$ of some finite degree $k$ such that
$\max_{(m,c)}
\left|f(L_{m,c})-p(m,c)\right|<\widetilde{\epsilon}$.
Let $\U{m,c}$ denote the upper envelope of the original distance functions
and let
$D^*=\min_{(m,c)} \U{m,c}>0$.
Since every distance function is approximated with an additive error at most
$\widetilde{\epsilon}$, the corresponding approximate upper envelope
$\widetilde{\mathcal{U}}(m,c)$ satisfies
$\left|\mathcal{U}(m,c)-\widetilde{\mathcal{U}}(m,c)\right|
\leq \widetilde{\epsilon}$
for every $(m,c)$. Let $(m^*,c^*)$ minimize $\mathcal{U}$ and let
$(\widetilde{m},\widetilde{c})$ minimize $\widetilde{\mathcal{U}}$. Then
\[
\begin{aligned}
\mathcal{U}(\widetilde{m},\widetilde{c})
\leq \widetilde{\mathcal{U}}(\widetilde{m},\widetilde{c})
+\widetilde{\epsilon}
\leq \widetilde{\mathcal{U}}(m^*,c^*)
+\widetilde{\epsilon} \leq \mathcal{U}(m^*,c^*)+2\widetilde{\epsilon} =D^*+2\widetilde{\epsilon}.
\end{aligned}
\]

Again, $D^*$ is bounded by a certain finite value, say the radius of the minimum enclosing circle of all the vertices. Thus, by choosing
$\widetilde{\epsilon}\leq\frac{\epsilon' D^*}{2}$,
we obtain
$\mathcal{U}(\widetilde{m},\widetilde{c})\leq (1+\epsilon')D^*$.
The case $D^*=0$ can be handled separately using the exact feasibility
conditions.

\subsubsection{Approximate \texorpdfstring{$(1,2)$}{(1,2)}-center}
Considering polynomial approximation of all $O(n^2r+nr^2)$ distance
functions, we construct the upper envelope of these approximated functions.
In \cite{Coreset} it was shown that the upper envelope of these functions
can be constructed with $(1+\epsilon_0)$-factor approximation in
$O(n^2r+nr^2+1/\epsilon_0^c)$ time for any $\epsilon_0>0$ and some
constant $c$. This leads us to state the following:
\begin{theorem}
    For curves in the plane, one can compute an $(1+\epsilon)$-factor approximation of
    $(1,2)$-center in
    $O(n^2r+nr^2+1/\epsilon^s)$ time for any $\epsilon>0$ and some
    constant $s$.
\end{theorem}

\subsection{A \texorpdfstring{$3$}{3}-approximation in time linear in the number of curves}
\label{sec_approx_linear}
Let $a$ be the center of the smallest enclosing ball of the start
vertices $\pi_1(0),\dots,\pi_r(0)$, and let $b$ be the center of the smallest
enclosing ball of the end vertices $\pi_1(n),\dots,\pi_r(n)$. We consider the segment $\overline{ab}$ as $(1,2)$-center.

\begin{algorithm}[H]
\DontPrintSemicolon
\caption{\textsc{Candidate-$(1,2)$-Center}}
\label{alg:candidate}
\KwIn{A set $\pi = \{\pi_1, \pi_2, \ldots, \pi_r\}$ of polygonal curves in $\mathbb{R}^d$.}
\KwOut{A line segment $\overline{ab}$ realizing a $3$-factor approximate $(1,2)$-center.}
\BlankLine
$a \gets$ the center of the smallest enclosing ball of all start vertices\;
$b \gets$ the center of the smallest enclosing ball of all end vertices\; \label{line:envelope}
\Return the line segment $\overline{ab}$\;
\end{algorithm}

\begin{theorem}
\label{thm_3approx}
$\overline{ab}$ produced by \cref{alg:candidate} is a $3$-factor approximation of the optimal $(1,2)$-center. For curves in the plane, it is computed in $O(r)$ time- linear in the number of curves and independent of the curve complexity.
\end{theorem}

\begin{proof}
The distance of the optimal $(1,2)$-center to the set of curves is $\Delta$.
Let $r_1,r_2$ be the radius of the smallest enclosing ball of the start and end vertices, respectively. If $r_1>\Delta$ then there does not exist any point $p$ for which $\max_i|\pi_i(0)-p|\le \Delta$ contradiction as start point $a'$ of optimal $(1,2)$-center satisfies $|\pi_i(0)-a'|\le \Delta$. So, $r_1\le \Delta$ and similarly $r_2\leq \Delta$; we say $r=\max(r_1,r_2)\le \Delta$. So, $|\pi_i(0)-a|,|\pi_i(n)-b|\le \Delta$ for all $i$. Let $d_i$ be the smallest possible \frd from $\pi_i$ to a line
segment; so $d_i\le\Delta$. In ~\cite{simplification_Agarwal} it has been shown that for any curve, $\delta_F\!\bigl(\pi_i,\overline{\pi_i(0)\pi_i(n)}\bigr)\le 2d_i$, and that two directed segments satisfy
$\delta_F(\overline{uv},\overline{xy})\le\max(|u-x|,|v-y|)$. So,
\[
\delta_F\!\bigl(\overline{\pi_i(0)\pi_i(n)},\overline{ab}\bigr)\le
\max\bigl(\|\pi_i(0)-a\|,\|\pi_i(n)-b\|\bigr)\le r\le\Delta .
\]

By the triangle inequality for the \frd,
\[
\delta_F(\pi_i,\overline{ab})\le
\delta_F\!\bigl(\pi_i,\overline{\pi_i(0)\pi_i(n)}\bigr)
+\delta_F\!\bigl(\overline{\pi_i(0)\pi_i(n)},\overline{ab}\bigr)\le
2d_i+\Delta\le3\Delta .
\]

The smallest enclosing circle of $r$ points in the plane is computed in $O(r)$
time~\cite{Megiddo}. This validates the statement of \cref{thm_3approx}.
\end{proof}

\section{Discussion}
We have shown that an optimal $(1,2)$-center for polygonal curves can be computed in polynomial time. However, the techniques developed here do not appear to extend directly to the computation of a $(1,3)$-center or, more generally, to the case where $k=1$ and $\ell$ takes larger values. An interesting open question is whether, while keeping the computational complexity manageable, the number of centers $k$ can be increased. A positive answer to this question would lead to effective clustering methods for polygonal curves.

\paragraph*{Acknowledgment :} The authors are grateful to Utsav Choudhury, Matthew (Matya) Katz, and Anil Maheshwari for stimulating discussions that greatly helped in writing this paper.

\bibliography{new}
\end{document}